\documentclass[aos,preprint]{imsart}

\RequirePackage[OT1]{fontenc}
\RequirePackage[leqno,cmex10]{amsmath}
\RequirePackage{amsthm}
\RequirePackage[numbers]{natbib}
\RequirePackage[colorlinks,citecolor=blue,urlcolor=blue]{hyperref}

\usepackage{mathrsfs,dsfont}
\usepackage{latexsym, stmaryrd}
\usepackage{amssymb}
\usepackage[english]{babel}
\usepackage[latin1]{inputenc}
\usepackage{color}
\usepackage{amsfonts}
\usepackage{bbm}
\usepackage{enumerate}
\RequirePackage{xspace}
\usepackage[ruled, section]{algorithm}
\usepackage{algpseudocode}
\usepackage{subfigure}
\usepackage {graphicx}
\usepackage {pgf}
\usepackage{yhmath}
\usepackage[normalem]{ulem}

\usepackage{multido}
\usepackage{pstricks,pst-plot,pstricks-add,pst-math}
\usepackage{mathtools}

\usepackage{enumerate}

\numberwithin{equation}{section}
\numberwithin{figure}{section}
\numberwithin{table}{section}
\theoremstyle{plain}
\newtheorem{thm}{Theorem}[section]
\newtheorem{lem}[thm]{Lemma}

\newtheorem{defn}[thm]{Definition}

\theoremstyle{definition}

\newcommand{\dB}{ \mathds B}
\newcommand{\R}{ \mathds R}
\newcommand{\dS}{ \mathds S}
\newcommand{\dU}{ \mathds U}

\newcommand{\dX}{ \mathds X}
\newcommand{\dY}{ \mathds Y}
\newcommand{\dV}{ \mathds V}
\newcommand{\dW}{ \mathds W}

\newcommand{\fd}{ \mathfrak d}
\newcommand{\fg}{ \mathfrak g}
\newcommand{\fS}{ \mathfrak S}
\newcommand{\gT}{ \boldsymbol\tau}
\newcommand{\gN}{ \boldsymbol\nu}
\newcommand{\gM}{ \boldsymbol\mu}

\theoremstyle{remark}
\newtheorem{rem}[thm]{Remark}
\newcommand{\ol}[1]{\left\langle #1\right\rangle}
\newcommand{\bkt}[1]{\left\llbracket #1\right\rrbracket}

\newcommand{\ul}[1]{\underline{ #1}}

\renewcommand{\le}{\leqslant}
\renewcommand{\ge}{\geqslant}

\newcommand{\eps}{\varepsilon}

\newcommand{\norm}[1]{\left\Vert#1\right\Vert}

\newcommand{\iid}{\emph{i.i.d.}~}

\let\ga=\alpha  \let\gc=\gamma  
\let\gf=\varphi    \let\gk=\kappa \let\gl=\lambda 
 \let\go=\omega   \let\gs=\sigma \let\gt=\tau 
  
\let\gC=\Gamma   \let\gL=\Lambda \let\gTh=\Theta
\let\gO=\Omega           

\newcommand{\cD}{\mathcal{D}}\newcommand{\cF}{\mathcal{F}}
\newcommand{\cI}{\mathcal{I}}
\newcommand{\cK}{\mathcal{K}}\newcommand{\cL}{\mathcal{L}}
\newcommand{\cM}{\mathcal{M}}

\newcommand{\cT}{\mathcal{T}}

\newcommand{\V}[1]{\ensuremath{\boldsymbol{#1}}\xspace}
\newcommand{\bre}[1]{\ensuremath{\breve{#1}}}

\newcommand{\vone}{\mathbf{1}}

\newcommand{\vA}{\mathbf{A}}

\newcommand{\vG}{\mathbf{G}}

\newcommand{\vP}{\mathbf{P}}\newcommand{\vQ}{\mathbf{Q}}

\newcommand{\vU}{\mathbf{U}}\newcommand{\vW}{\mathbf{W}}
\newcommand{\vX}{\mathbf{X}}

\newcommand{\dN}{\mathds{N}}

\newcommand{\bsP}{\pmb{\mathscr{P}}}
\newcommand{\bsQ}{\pmb{\mathscr{Q}}}
\newcommand{\bsA}{\pmb{\mathscr{A}}}

\newcommand{\sH}{\mathscr{H}}

\DeclareMathOperator{\E}{\mathds{E}}
\DeclareMathOperator{\pr}{\mathds{P}}

\newcommand{\sM}{\mathscr{M}}

\def\beq{ \begin{equation} }
 \def\eeq{ \end{equation} }
 \def\beqx{ \begin{equation*} }
 \def\eeqx{ \end{equation*} }
 \def\beqa{\begin{eqnarray}}
 \def\eeqa{\end{eqnarray}}
 \def\beqax{\begin{eqnarray*}}
 \def\eeqax{\end{eqnarray*}}

\newcommand{\til}{\tilde}

\usepackage[T3,T1]{fontenc}
\DeclareSymbolFont{tipa}{T3}{cmr}{m}{n}
\DeclareMathAccent{\inv}{\mathalpha}{tipa}{16}

\begin{document}

\begin{frontmatter}
\title{{\small Consistent~detection~and~optimal~localization of all detectable change-points in piecewise stationary arbitrarily sparse network-sequences}}
\runtitle{Change Point Detection for Network-valued Time-series}

\begin{aug}

\author{\fnms{Sharmodeep} \snm{Bhattacharyya}\thanksref{m3}\ead[label=e1]{bhattash@science.oregonstate.edu}}
\and
\author{\fnms{Shirshendu} \snm{Chatterjee}\thanksref{m1,t3}\ead[label=e2]{shirshendu@ccny.cuny.edu}}
\and
\author{\fnms{Soumendu Sundar} \snm{Mukherjee}\thanksref{m2, t1}\ead[label=e3]{soumendu041@gmail.com}}

\thankstext{t3}{Supported in part by PSC-CUNY and Simons Foundation}
\thankstext{t1}{Most of this work appeared in the PhD thesis \cite{mukherjee2018thesis} of the SSM. SSM is supported by an INSPIRE Faculty Fellowship from the Department of Science and Technology, Government of India}
\runauthor{Mukherjee, Bhattacharyya and Chatterjee}

\affiliation{City University of New York\thanksmark{m1}}
\affiliation{Indian Statistical Institute, Kolkata\thanksmark{m2}}
\affiliation{Oregon State University\thanksmark{m3}}

\address{Interdisciplinary Statistical Research Unit\\
Indian Statistical Institute\\
Kolkata, West Bengal 700108\\
India\\
\printead{e3}}
\address{Department of Statistics\\
237 Weniger Hall\\
Corvallis, OR, 97331\\
\printead{e1}}
\address{Department of Mathematics\\
North Academic Center 8/133\\
160 Convent Ave\\
New York, NY, 10031\\
\printead{e2}}

\end{aug}

\begin{abstract} 
\indent 
We consider the offline change point detection and localization problem in the context of piecewise stationary networks,
where the observable is a finite sequence of networks. We develop
algorithms involving some suitably modified CUSUM statistics based on 
adaptively trimmed adjacency matrices of the observed
networks for both detection and localization of single or multiple
change points present in the input data. We provide rigorous theoretical analysis and finite sample estimates evaluating the performance
of the proposed methods when the input (finite sequence of networks)
is generated from an inhomogeneous random graph model, where
the change points are characterized by the change in the mean adjacency matrix. We show that the proposed algorithms can detect (resp. localize) all change points, where the change in the expected adjacency matrix is above the minimax detectability (resp. localizability)
threshold, consistently without any a priori assumption about (a) a
lower bound for the sparsity of the underlying networks, (b) an upper
bound for the number of change points, and (c) a lower bound for
the separation between successive change points, provided either the
minimum separation between successive pairs of change points or the
average degree of the underlying networks goes to infinity arbitrarily
slowly. We also prove that the above condition is necessary to have
consistency. Finally, we evaluate the performance and complexity of our methods empirically using
simulated data sets, and demonstrate
the superiority of our algorithms over relevant existing approaches.
\end{abstract}


\end{frontmatter}

\section{Introduction}
\label{sec:intro}
As network data sets have grown in complexity in the recent decades, so has the prevalence of temporal or time-varying or time-series of networks. Time series of networks have emerged in several different fields of study. Examples of temporal or time series of networks in different fields of study include time-series of social networks \cite{panisson2013fingerprinting, stopczynski2014measuring, rocha2010information}, epidemiological networks \cite{salathe2010high, rocha2011simulated}, animal networks \cite{gates2015controlling, lahiri2007structure},  mobile and online communication networks \cite{krings2012effects, ferraz2015rsc, jacobs2015assembling}, economic networks \cite{popovic2014extraction, zhang2014dynamic}, brain networks \cite{park2013structural, sporns2013structure}, genetic networks \cite{rigbolt2011system} and ecological networks \cite{blonder2012temporal}, to name a few. Analysis of temporal networks in terms of modeling, statistical behavior, dynamics, community detection and change point detection has been investigated in several recent works (see \cite{holme2012temporal, holme2015modern, peixoto2015inferring, sikdar2016time, peixoto2018change} for some review of recent works). In this paper, we shall concentrate on the problem of change point detection for time-series of networks.

Change point detection is a classical problem in statistics going all the way back to the early days of statistical quality control \cite{page1954continuous,page1957problems,girshick1952bayes}. However, the problem of change point detection has gained significant importance and applicability in many fields such as medical diagnostics \cite{yang2006adaptive, staudacher2005new, bosc2003automatic}, gene expression \cite{picard2005statistical, hocking2013learning}, online activity \cite{levy2009detection}, speech and image analysis \cite{harchaoui2009regularized, radke2005image}, climate science \cite{reeves2007review}, finance \cite{bai1998estimating, lavielle2007adaptive} and many more. The problem of change point detection started with detection of change in the mean of normal model \cite{page1954continuous} but since has also been generalized to many different data types such as time-series data \cite{aminikhanghahi2017survey} and multivariate data \cite{brodsky2013nonparametric} as well for detecting change in different parameters of the data distribution such as variance, correlation, density and so on. 

The change point detection problem can be broadly classified into two types.
\begin{enumerate}
	\item \emph{Offline change point detection:} In this type of problem, the whole data sequence is available and the change points are detected within the data sequence. This problem was studied in the beginning by Page (1954) \cite{page1954continuous} and Girshick and Rubin (1952) \cite{girshick1952bayes}.
	\item \emph{Online change point detection:} In this type of problem, the data is available sequentially and the change points are detected based on the available data. The online version was initially studied by Kolmogorov (1950), Shiyarev (1963) \cite{shiryaev1963optimum}, Lorden (1971) \cite{lorden1971procedures} and others.
\end{enumerate}
There is a huge literature on the univariate change point problem and possible solutions. An excellent treatment can be found in the book \cite{brodsky2013nonparametric}. The multivariate versions of the problem are significantly more complex. Some notable works are \cite{zhang2010detecting,siegmund2011detecting,srivastava1986likelihood,james1992asymptotic} in the parametric setting, and \cite{harchaoui2009kernel,lung2011homogeneity,chen2015} in the non-parametric setting.

In this article, we tackle the problem of change point detection in temporal network data, that is one observes a series of networks indexed by time and wishes the check if there is a time-point (so-called change point) when there is a significant change in the structure of these networks. Potential applications are in, for instance, brain imaging, where one has brain scans of individuals collected over time and is looking for abnormalities, ecological networks observed over time, where one wonders if there is a structural change, and so on. The classical CUSUM statistic \citep{page1957problems} for univariate change point problems can be used in the network problem as well, and provides a unified way of constructing estimates of change points. It is also amenable to theoretical analysis because of the averaging structure present. In this paper, we will investigate its theoretical properties in a quite general setup. We stress here that we observe the whole time series ahead of our analysis, this is thus an \emph{offline} or \emph{a posteriori} change point problem. We will not discuss the online version of the problem here, which is also quite interesting.

There has been some recent works on the problem of network change points. For example, \cite{peel2015detecting} postulate a hierarchical random graph model and use a Bayesian procedure to detect change points. \cite{park2013anomaly} use local graph statistics for change point and anomaly detection in dynamic networks. For a survey of techniques used in the related problem of anomaly detection in graphs, see \cite{ranshous2015anomaly}. Most two sample graph tests can be used for the change point problem viewed as a multiple testing problem. For example, an eigenvalue based test for the ER vs SBM problem is worked out in \cite{cape2017kato}. Minimax lower bounds for two sample tests for inhomogeneous Erd\"os-R\'enyi graphs have been derived in \cite{ghoshdastidar2017two}. As we will see later, such lower bounds are closely related to the lower bounds we derive for the change-point problem we consider. Although much empirical work has been done, not much theory can be found, and most theoretical results focus on particular structures or specialized models. An exception is \cite{roy2017change}, where the authors \cite{roy2017change} model networks as a Markov random field and estimate the change point using a penalized pseudo-likelihood and prove its consistency at a near classical (i.e. fixed dimensional) rate under a restricted strong convexity type assumption on the log-pseudo-likelihood. Although their results are in a high-dimensional setting, and allow more complicated node interaction than random graphs with independent edges, the role of network sparsity in their setup is not clear. 

Some recent works \citep{wang2018optimal, bhattacharjee2018change, zhao2019change} propose methods for change point detection in networks generated from block models and graphon models with some theoretical results on the consistency of the detection methods. Graphon estimation based methods \citep{zhao2019change} only work for dense graphs. On the other hand, \citep{bhattacharjee2018change} only consider block-models. The work most closely related to the current paper is \cite{wang2018optimal}. Although they propose an algorithm which is almost minimax optimal, their algorithm requires two independent copies of the network time-series, which severely limits its usefulness. In contrast, we provide two efficient algorithms that are minimax optimal and do not require such restrictive assumptions. Our minimax lower bounds are also much more precise than theirs, including provisions for perturbations of specific ranks.


\subsection{Problem Description} 
\label{sec_intro_problem}
In this section, we formally represent the inference problem that we consider in this paper. Let us consider that the number of layers in the network sequence is given by $T$ ($T \in \dN$), the number of nodes in each network layer is represented by $n$ ($n \in\dN$), and number of change points in the network sequence is given by $K$ ($0\le K<T$). Let us denote, $[n]:=\{1, 2, \ldots, n\}$, and $\dS_n$ be the set of all $n\times n$ symmetric  matrices having entries in $[0,1]$ and zeros on the diagonal. The set of possible change point locations is given by 
\begin{align*}
\dU(K,T) := \{\gT=(\tau_0, \tau_1, \ldots, \tau_{K+1}): 0=\tau_0<\tau_1<\cdots<\tau_K<\tau_{K+1}=T\}, 
\end{align*} 
the set of possible expected adjacency matrices is given by $\dS_n^T:=\{\bsQ=(\vQ^{(1)}, \ldots, \vQ^{(T)}): \vQ^{(t)}\in\dS_n,\ \forall  t\in[T]\}$. For $\bsQ \in\dS_n^T$, let
$\cK(\bsQ) := |\{t\in [T-1]:   \vQ^{(t)} \ne \vQ^{(t+1)}\}|$ be the number of change points in $\bsQ$, and $\cT(\bsQ)\in \dU(\cK(\bsQ),T)$ be the sequence of all the change points of $\bsQ$.  Let the no change point situation is represented by $\dS_{n,0}^T:=\{\bsQ\in\dS_n^T: \cK(\bsQ)= 0\}$, and at least one change point is presented by $\dS_{n,+}^T:=\dS_n^T\setminus \dS_{n,0}^T$.

Now, we describe the statistical hypothesis testing problem in the context of  detecting of change points in a finite sequence of networks. The observable is a sequence $(G_n^{(1)}, \ldots, G_n^{(T)})$, which is represented by the corresponding sequence of adjacency matrices $\bsA:(\vA^{(1)}, \ldots, \vA^{(T)})\in\dS_n^T$,  consisting of $T$ simple undirected networks on the same node set $\{v_1,  \ldots, v_n\}$ having $n$ nodes.  $A^{(t)}_{ij}$ equals 1 (resp.~0) if nodes $v_i$ and $v_j $ are (resp.~not) linked in the $t$-th snapshot $G_n^{(t)}$. Let  $\bsP:=(\vP^{(1)}, \ldots, \vP^{(T)})$, where $\vP^{(t)}:=\E\vA^{(t)}$. So, $\bsP \in \dS_n^T$. We define the hypothesis test in terms of the following pairs of competing composite hypotheses
For $\dX\subset\dS_{n,0}^T$ and $\dY\subset\dS_{n,+}^T$, consider the
\[ H_0:\;\bsP\in\dX; H_1:\; \bsP\in\dY.  
\]
In plain English, we want to test whether the sequence of the mean adjacency matrices belongs to a set $\dX$ having no change point, or to a set $\dY$ having at least one change point. 

In other words, the question is when can one decide reliably whether there is a sequence of change points in $\bsP$ or not. The answer depends on the criterion used for judging the performances of the decision rules. There are mainly two paradigms in statistical decision theory, namely the Bayesian and the minimax approach. We will consider the second approach in this paper. Recall that a nonrandomized test $\gO_{n,T}$ is a measurable function of the observable $\bsA$ taking values in $\{0, 1\}$. The minimax risk of such a test $\gO_{n,T}$ for the hypotheses $H^\dX_0$ and $H^\dY_1$ is
\begin{align}\label{minimax_risk}
\Pi(\gO_{n,T}; \dX,\dY) := 
\overbrace{\sup_{\bsQ\in\dX}\pr_{\bsQ}(\gO_{n,T}=1)}^{Type\ I\ error}+\overbrace{\sup_{\bsQ\in\dY}\pr_{\bsQ}(\gO_{n,T}=0)}^{Type\ II\ error}.
\end{align}
Here and later $\pr_{\bsQ}$  represents the distribution of $\bsA$ when $\E\bsA=\bsQ$. Recall the following definitions for the asymptotic properties of hypothesis test.
\begin{defn}
	A  test $\gO_{n,T}$ is called \emph{asymptotically powerful} for the hypotheses $H_0: \bsP\in\dX$ and $H_1:\bsP\in\dY$, if $\limsup_{nT\to\infty}\Pi(\gO_{n,T};\dX,\dY)$ equals $0$.
	
	A  test $\gO_{n,T}$ is called \emph{asymptotically powerless} for the hypotheses $H_0: \bsP\in\dX$ and $H_1:\bsP\in\dY$, if $\liminf_{nT\to\infty} \Pi(\gO_{n,T};\dX,\dY)$ is at least $1$.
\end{defn}
\begin{defn}[Detectability]
	For any $\dX\in\dS_{n,0}^T$ and $\dY\in\dS_{n,+}^T$, the hypotheses $H_0: \bsP\in\dX$ and $H_1: \bsP\in\dY$ are called 
	\begin{itemize}
		\item {\it consistently distinguishable } if there is an asymptotically powerful test for them.
		\item {\it consistently indistinguishable } if there is no asymptotically powerful test for them.
		\item {\it asymptotically indistinguishable } if all tests for them are asymptotically powerless.
	\end{itemize}
\end{defn}
Our first result (see \textsection \ref{sec_window_theory}) in this paper is to obtain the maximum detectable set. 

Having considered the problem of detectability, next we consider the problem of localizability.  Consider the collection $\cup_{K\in[T-1]} \dU(K,T)$ of sequences of probability matrices having at least one change point.
\begin{defn}
	For $\gL\le T$, a set $\dY\subset\{\bsQ\in\dS_n^T: \cK(\bsQ)\ge 1\}$ is called $\gL$-localizable, if there is an estimator $\breve\gT$ of change points and a constant $c>0$ such that 
	\[ \lim_{nT\to\infty}\inf_{\bsQ\in\dY}\pr_{\bsQ}\left(\cK(\breve\gT)=\cK(\bsQ) \text{ and } \cT_k(\bsQ)  \in \breve\tau_k\pm c\gL \text{ for all } k\in[\cK(\gT)]\right) = 1. \]
\end{defn}
In our second result (see \textsection \ref{sec_window_theory} and \textsection \ref{sec:wbs}), we have determined which subsets of $\dS_n^T$ are  $\gL$-localizable.

\subsection{Our results}
\label{sec_intro_results}
To determine detectability of a given change point detection problem, one needs to measure the amount  of change among the expected observables. In the current case, there are three main quantitative variables that determine the complexity of a change point detection problem, namely  (a) the gaps between successive change points,  (b) the differences of the probability matrices at the change points, and (c) the level of sparsity of the probability matrices.
Researchers have used one of the standard matrix norms for quantifying (b). In this paper, we will use the spectral norm for (b) (denoted by $||\cdot||$). For $\bsQ\in\dS_n^T$, define the three quantities as:
\begin{align*}
(a)\ & \text{Cushion: }  \fg(\bsQ) := \min_{k\in[\cK(\bsQ)+1]}(\cT_k(\bsQ)-\cT_{k-1}(\bsQ)),\\
(b)\ & \text{Signal: } \fS(\bsQ) := \begin{cases}
\min_{k\in[\cK(\bsQ)]} \norm{\vQ^{(\tau_k+1)}-\vQ^{(\tau_k)}} & \text{ if } \cK(\bsQ)\ge 1\\
0 & \text{ if } \cK(\bsQ)=0\end{cases},\\
(c)\ & \text{Sparsity: }  \cD(\bsQ):=\max_{i,j\in[n], t\in[T]} Q^{(t)}_{ij}.
\end{align*}
Also let $\underline{S}_n^T$ be the maximum subset of $\dS_n^T$ such that
\[ \inf_{\bsQ\in\underline{S}_n^T} \fg(\bsQ)\cD(\bsQ) \in \go(1) \text{ as $n$ or $T$ or both goes to $\infty$.}\]
\begin{itemize}
	\item We have shown (see Algorithm 1 for an asymptotically powerful test procedure and Theorem \ref{multiple tau detect} for the precise statement) that if
	\[ \dY=\left\{\bsQ\in\dS_n^T: \fS(\bsQ)\gtrsim \sqrt{\frac{\cD(\bsQ)}{\fg(\bsQ)}, \bsQ\in \underline{S}_n^T} \right\},\]
	\[ \dV=\left\{\bsQ\in\dS_n^T: \fS(\bsQ)\lesssim \sqrt{\frac{\cD(\bsQ)}{\fg(\bsQ)} \text{ or } \bsQ\not\in \underline{S}_n^T} \right\},\]
	$\tilde{\dV}$ is the analogue of $\dV$ with ``$\lesssim$" replaced by $\ll$, and $\vX=\dS_{n,0}^T$, then 
	\begin{align*}
	H_0: \bsP\in\dX \text{ and } H_1: \bsP\in\dY \text{ are  consistently distinguishable} , \\
	H_0: \bsP\in\dX \text{ and } H_1: \bsP\in\dV \text{ are  consistently indistinguishable} , \\
	H_0: \bsP\in\dX \text{ and } H_1: \bsP\in\tilde{\dV} \text{ are  asymptotically indistinguishable} 
	\end{align*}
	
	\item We have also shown that for any $\gL\le T$, 
	\[ \left\{\bsQ\in\dS_n^T: \fS(\bsQ)\gtrsim \sqrt{\frac{\cD(\bsQ)}{\gL\wedge\fg(\bsQ)}, \bsQ\in \underline{S}_n^T} \right\}\]
	is $\gL'$-localizable for all $\gL'\ge\gL$.
	Thus, if we restrict our consideration to $\{\vQ\in\dS_n^T: \cD(\vQ)=d_0, \fg(\vQ)=\gk_0\}$, then the detectability threshold is $\sqrt{d_0/\gk_0}$ if $d_0\gk_0\in\go(1)$. See Theorem \ref{multiple tau local} for the precise statement
\end{itemize}

\subsection{Outline of the paper}
\label{sec_outline}
	The remainder of the paper is organized as follows. In \textsection \ref{sec:methods}, we describe the change point detection algorithms proposed in the paper. In \textsection \ref{sec:theory}, we state the theoretical results regarding the performance of the proposed change point detection algorithms and the detectability threshold of change points for a large class of probability distributions. 

\section{Change Point Detection Methods}
\label{sec:methods}
\subsection{Notations}
\label{sec_notation}

Let $[n] := \{1, 2, \ldots, n\}$ for $n\in\dN$, $\sM_{m,n}$ be the set of all $m\times n$ matrices which have exactly one 1 and $n-1$ 0's in each row. $\R^{m\times n}$ denotes the set of all $m\times n$ real matrices. $||\cdot||_2$ is used to denote Euclidean $\ell_2$-norm for vectors in $\R^{m\times 1}$. $||\cdot||$ is the spectral norm on $\R^{m\times n}$. $||\cdot||_F$ is the Frobenius norm on $\R^{m\times n}$, namely $||M||_F := \sqrt{trace(M^T M)}$. $\vone_{m} \in \R^{m\times 1}$ consists of all 1's, $\mathbf 1_A$ denotes the indicator function of the event $A$. 
If $\vA\in\R^{m\times n}$,  $I\subset [m]$ and $j\in [n]$, then $\vA_{I,j}$ (resp.~$\vA_{I,*}$)  denotes the submatrix  of $\vA$ corresponding to row index set $I$ and column index $j$ (resp.~index set $[n]$).  $\gl_i(\vW), i\in[n],$ will denote the $i$-th largest eigenvalue of $\vW\in\R^{n\times n}$.
$\bkt{X}:=X-\E(X)$ for any random variable or random matrix $X$.

\subsection{Network Data and Model}
\label{sec:algo}
We consider the setup where one observes a sequence of $T$ networks, $()G_n^{(t)})_{t=1}^T$, with adjacency matrices $(\vA^{(1)}, \vA^{(2)}, \ldots, \vA^{(T)})$ on the same set of nodes $\{v_1, v_2, \ldots, v_n\}$. For each $t\in\{1, \ldots, T\}$, the $t$-th network $G_n^{(t)}$ is represented by the corresponding adjacency matrix $\vA^{(t)}_{n\times n}$ whose elements are $A^{(t)}_{ij}\in\{0,1\}$. $A^{(t)}_{ij} = 1$ if node $v_i$ is linked to node $v_j$ at time $t$, and $A^{(t)}_{ij} = 0$ otherwise. Thus, the numerical data for the community detection problem consists of $T\ge 1$ adjacency matrices $\left(\vA^{(1)}_{n\times n}, \ldots, \vA^{(T)}_{n\times n}\right)$. We shall only consider undirected and unweighted graphs in this paper. However, the conclusions of the paper can be extended to positively weighted graphs with non-random weights in a quite straightforward way by considering weighted adjacency matrices. The theoretical analysis in this paper can be easily extended to positively weighted adjacency matrices.

We consider that the set of adjacency matrices are generated independently from an \emph{inhomogeneous random graph model}. 
\begin{defn}[Multilayer Inhomogeneous Random Graph Model (MIRGraM)] \label{defn_mirgram}
	A sequence of $T$ ($T\in \dN$) symmetric adjacency matrices $(\vA^{(1)}, \vA^{(2)}, \ldots, \vA^{(T)})$ follows \emph{Multilayer Inhomogeneous Random Graph Model (MIRGraM)} with parameters $\bsP = (\vP^{(1)}, \ldots, \vP^{(T)})$, each of size $n\times n$, if,
	\begin{equation}
	\label{eq_irgm}
	A^{(t)}_{ij} \stackrel{iid}{\sim} \text{Ber}(P^{(t)}_{ij}). 
	\end{equation}
	for $i > j$ with $i,j=1, \ldots, n$ and $t=1, \ldots, T$. Lastly, $\vA^{(t)} = (\vA^{(t)})^T$ and $\text{diag}(\vA^{(t)}) = \V{0}$ for each $t=1, \ldots, T$.
\end{defn}

We consider the setup where, $\bsP$ changes at certain instances, like $(\tau_1, \ldots, \tau_{K})$, where $1 = \tau_0 < \tau_1 < \cdots < \tau_K < \tau_{K+1} = T$. Our goal becomes estimating the instances of change-point $(\tau_1, \ldots, \tau_K)$. We consider the problem of \emph{multiple change-point detection and localization.}

In this problem we consider that $K \geq 1$, so we have to estimate multiple change-points $(\tau_1, \ldots, \tau_K)$ both in form of a point estimator (\emph{detection}) and interval estimator (\emph{localization}). We propose two different algorithms to address multiple change-point detection problem.
	\begin{enumerate}[(i)]
		\item We develop a window-based algorithm (Algorithm 1) to detect and localize the multiple change-points. The formal setup and algorithm is given in \textsection \ref{sec_multiple}.
		\item We develop a wild binary segmentation algorithm (Algorithm 2) motivated by \cite{fryzlewicz2014wild}. The formal setup and algorithm is given in \textsection \ref{sec_wbs}.
	\end{enumerate}
We will now go into the details of algorithms and the setup.

\subsection{Multiple change-point detection and localization}
\label{sec_multiple}
The problem of multiple change-point detection also starts with a sequence of networks represented by the symmetric adjacency matrices $\left(\vA^{(1)}_{n\times n}, \ldots, \vA^{(T)}_{n\times n}\right)$. The model under a multiple change-point with parameters $\bsP = \left(\vP^{(1)}, \ldots, \vP^{(T)}\right)$ can be framed as -
\begin{enumerate}[(i)]
	\item if $\vP^{(t)} = \vP$, for all $t \in [T]$, there exists no change-point;
	\item if there exists a set of change-points $(\tau_1, \tau_2, \ldots, \tau_K)$ with $\tau_0 := 0  \le \tau_1 \le \cdots \le \tau_K \le T =: \tau_{K+1}$ such that 
	\[\vP^{(t)} = \vQ_k \ \  \text{for } \tau_{k-1} \le t \le \tau_k,\ k\in [K]\]
	where, $\vQ_k \in \dS_n$ for $k\in [K]$, then there exists a multiple change-points $(\tau_1, \tau_2, \ldots, \tau_K)$. 
	\item We additionally assume  $\fg(\bsP)$ represents the \emph{minimum cushion} at the boundaries of the network sequence as well as the \emph{minimum cushion} between successive change points, thus satisfying the relationship
	\beq\label{cushion_def}
	\fg(\bsP) := \min_{0\le k\le K} (\tau_{k+1}-\tau_k)
	\eeq
\end{enumerate}
\emph{The goal then becomes declaring absence of change-point when case (i) is true and estimating $(\tau_1, \tau_2, \ldots, \tau_K)$ when the case (ii) is true.}

Let $\bar D =\frac{1}{nT}\sum_{i,j\in[n], s\in[T]} \vA^{(s)}_{i,j}$ be the sample average degree of a node over all layers and for some constant $\mu > 0$, define $\ul\fd_\mu:=\frac{\bar D}{1+\sqrt{4\mu}}$ as a normalized version of the average degree, $\bar D$. The degree of each of the vertices average over all layers is given by
\[ D_{i} = \frac{1}{T} \sum_{j \in [n], s\in [T]} \vA^{(s)}_{i,j}. \ \ \ \text{for } i \in [n].\] 
Define the function $\gf(z):=z\log(z)-z+1$~($\gf:[1,\infty)\mapsto[0,\infty))$. We need to define the following population quantities based on the parameters of MIRGraM, which are given in Definition \ref{defn_mirgram},
\begin{align}
d:=n\left(\max_{i,j\in[n],t\in[T]}\E A^{(t)}_{ij}\right), \quad \fd:=\max_{i\in[n],t\in[T]}\sum_{j\in[n]}\E A^{(t)}_{ij}. \label{d def} 
\end{align}
Let $\cL$ be a collection of intervals, each of length $\gL$. Let $\gk$ be the cushion at the boundary of each interval in $\cL$. For any $\ell\in[|\cL|]$, let $T_\ell$ be the set of left end-points of intervals in $\cL$. Then, let us define,
\begin{align}
\left[\begin{array}{c}
\eps_\mu\\ \Psi_\mu\\ \eta\end{array}\right] :=\begin{cases}
\left[\frac 16 , \gf^{-1}\left(\frac{(2\log|\cL|)\vee (3(\gL\wedge\gk)\fd)}{(\frac 13-\eps_\mu)(\gL\wedge\gk)\fd}\right), 1\right]^T,\text{if $|\cL|\not\in (e, e^n)$, } 
\\
\left[\frac{2\eta-2}{6\eta-3},  \gf^{-1}\left(\frac{\frac{\log\log( |\cL|)}{(\gL\wedge\gk)\fd}\vee 3}{\frac 23 - 2\eps_\mu}\right), \frac{\log(n)}{\log\log|\cL|}\right]^T,
\text{otherwise,} 
\end{cases}\label{Psi def_gen}\\
\gC:=\begin{cases}
\left\lceil \frac{25n}{4}\exp\left(-\frac{3\eps_\mu}{2-6\eps_\mu} \left[\left(\log\log|\cL|\right)\vee\frac{3(\gL\wedge\gk)\bar D}{1+\sqrt{4\mu}}\right]\right)\right\rceil, \text{ if } |\cL|\in(e, e^n), \\
\left\lceil \frac{25n}{4}\exp\left(-\frac{3\eps_\mu}{1-3\eps_\mu}\left[\left(2\log|\cL|\right)\vee \frac{3(\gL\wedge\gk)\bar D}{1+\sqrt{4\mu}}\right] \right)\right\rceil, \text{otherwise}.
\end{cases} \label{gC def_gen}
\end{align}

A natural statistic is based on the cumulative averages (\emph{cusum}) of estimates of the $\vP^{(t)}$'s. Such cusum statistics are widely used in change-point detection problems \cite{brodsky2013nonparametric}. A natural estimate of $\vP^{(t)}$ is $\vA^{(t)}$ for any $t\in [T]$. However for sparse networks, $\vA^{(t)}$ is not a good enough estimate of $\vP^{(t)}$ under operator norm. Hence given $\gL$ and $\gk$, for each interval $\cI \in \cL$, we obtain submatrices $\tilde \vA^{(t)}$ of $\vA^{(t)}$ for each $t \in [T]$ by removing some high degree vertices using the threshold $\gC$, where $\gC$ is defined in \eqref{gC def_gen}. The thresholding uses the ordered vertices in terms of degree of each vertex, $D_i$, such that $D_{(1)} \leq D_{(2)} \leq \cdots \leq D_{(n)}$, and vertices corresponding to $D_{(j)}$, such that $j > (n+1-\gC)$, are pruned. The thresholded matrices $(\tilde \vA^{(t)})_{t=1}^T$ are better approximations of $(\vP^{(t)})_{t=1}^T$ in terms of operator norm
 We use $(\tilde \vA^{(t)})_{t=1}^T$ to construct the cusum statistics, $\vG^{(t)}$, for interval $\cI_{\ell} \in \cL$,
\begin{align}\label{Gt_def}
\vG^{(t)} := \sqrt{\frac{t}{\gL}\left(1 - \frac{t}{\gL}\right)} \left(\frac{1}{t} \sum_{i = T_{\ell}+1}^t \tilde\vA^{(i)}- \frac{1}{\gL - t} \sum_{i = T_{\ell}+t+1}^{T_{\ell}+\gL} \tilde\vA^{(i)}\right) 
\end{align}
$\text{ for } T_{\ell}+\lfloor\gk/3+1\rfloor\le t\le T_{\ell}+\gL-\lfloor\gk/3\rfloor.$

\subsubsection{Window-based Algorithm}
\label{sec_multiple_window}
In order to build the \emph{window-based algorithm}, Algorithm 1, for multiple change-point detection, we consider the set of intervals -
\begin{align}
\label{sJ_ell def}
\cL := \{\cI_{\ell}\} \text{ where, } \cI_{\ell}:= (T_{\ell}, (T_{\ell}+\gL) \wedge T],  \text{ and } \\
T_{\ell} = (\ell - 1)\lfloor\gL/3\rfloor \text{ for } \ell \in \{1, 2, \ldots, \lceil3T/\gL\rceil -2\}. \nonumber
\end{align}
%
Along with the set of intervals $\cL$, there are other components of Algorithm 1.
\begin{enumerate}
	\item Cushion $\gk$ and interval length $\gL$ are two tuning parameters of Algorithm 1. For fixed $\gk$ and $\gL$, we consider the set of intervals $\cL$ as defined in \eqref{sJ_ell def}. In Algorithm 1, we consider that $\gL \leq \gk$, so, one can vary $\gL$ from $3$ to $\gk$ for a given $\gk$ to find the intervals of shortest length with a change point. 
	\item For each $\cI_\ell \in \cL$, $\norm{\vG^{(t)}}$ is minimized to get a candidate of change-point estimate, $\til \tau$.
	\item If $\norm{\vG^{(t)}}$ is greater than a data-dependent threshold as defined in Algorithm 1, then $\til \tau$ is change-point within localized interval $\cI_{\ell}$, otherwise not.
\end{enumerate}

\vspace{0.2in}
\framebox[\textwidth]
{\centering\parbox{.95\textwidth}{
		\textbf{Algorithm 1:} Window-based Change Point Detection \\
		\textbf{Input:} Adjacency matrices $\vA^{(1)}, \vA^{(2)}, \ldots, \vA^{(T)}$; cushion $\gk$.\\
		\textbf{Output:} Change point estimates with localized interval lengths and intervals $(\hat\tau, \gL, \cI)$. 
		\begin{enumerate}
			\item[] For $\gL = \gk, \gk-1, \ldots, 3$ do
			\item[] \quad Obtain $\bar D=\frac{1}{nT}
			\sum_{i,j\in[n], s\in[T]}A^{(s)}_{ij}$.
			\item[] \quad Obtain $\cL$ as defined in equation \eqref{sJ_ell def}.
			\item[] \quad Obtain $\gC$ as defined in equation \eqref{gC def_gen}.
			\item[] \quad For $\ell = 1, 2, \ldots, \lceil3T/\gL\rceil -2$ do
			\item[] \quad \quad Define $T_{\ell}$ and $\cI_{\ell}$ as in \eqref{sJ_ell def}. 
			\item[] \quad \quad For $i=1, 2, \ldots, n$ do
			\item[] \quad \quad \quad Obtain $D_{i}$.
			\item[] \quad \quad Order the values $D_{1}, \ldots, D_{n}$ to get 
			$D_{(1)}\le \cdots\le D_{(n)}$.
			\item[] \quad \quad Obtain row indices $i_1, \ldots, i_\gC$ such that $D_{i_k}\ge D_{(n+1-\gC)}$
			\item[] \quad \quad Obtain $\tilde \vA^{(s)}$ from $\vA^{(s)}$ for each $s\in\cI_{\ell}$ by removing rows and columns with indices $i_1, \ldots, i_\gC$. 
			\item[] \quad \quad For $t=T_{\ell}+\lfloor\gL/3\rfloor+1, \ldots, T_{\ell}+\gL-\lfloor\gL/3\rfloor$ do
			\item[] \quad \quad \quad Obtain $\vG^{(t)}$ as in \eqref{Gt_def}.
			\item[] \quad \quad Obtain $u=\arg\max_{t\in(T_{\ell}+\lfloor\gL/3\rfloor, T_{\ell}+\gL-\lfloor\gL/3\rfloor]}\norm{\vG^{(t)}}$.
			\item[] \quad \quad If $\norm{\vG^{(u)}}>\Theta_\mu\left[\frac{\bar D}{\gL}
			\left(\zeta+6+\frac{\log(|\cL|)}{\log(n)}\right)\right]^{1/2}$, declare $u$ as a change point in interval $\cI_{\ell}$ of length $\gL$. 
			\item[] Return the detected change-points corresponding to all interval lengths $\gL$ and interval $\cI_{\ell}$ as $(u, \gL, \cI_{\ell})$.
		\end{enumerate}
}}\\

\subsubsection{Wild Binary Segmentation}
\label{sec_wbs}
Wild binary segmentation is a randomized algorithm for multiple change-point detection \cite{fryzlewicz2014wild}. Let us denote $\cF_T^M$ as a set of $M$ random intervals
$[s_m, e_m]$, $m = 1, \ldots,M$, whose start and end points have been drawn (independently with replacement) uniformly from the set $\{1, \ldots, T\}$ with the property $e_m > s_m$. Note that,
\emph{wild binary segmentation} algorithm \cite{fryzlewicz2014wild} is a recursive one and consequently Algorithm 2 has also been written in a recursive format.
The main components of Algorithm 3 are given below.
\begin{enumerate}
	\item In a specific iteration, Algorithm 2 operates on an interval $(s, e)$ (where, $(s, e) \subseteq [1, T]$). Let $\cM_{s, e}$ be the set of indices, $m \in [M]$, of intervals in $\cF_T^M$, such that, $(s_m, e_m) \subseteq (s, e)$. 
	\item Cushion $\gk$ is the tuning parameter of the algorithm. We consider $\gk$ is given to the algorithm at any specific iteration.
	\item For each $\cI = (s_m, e_m)$ such that $m \in \cM_{s, e}$, say $\norm{\vG^{(t)}}$ is minimized at $t = u_m$ and if $\norm{\vG^{(u_m)}}$ is greater than a data-dependent threshold as defined in Algorithm 2, then $u_m$ becomes a candidate of change-point estimate.
	\item The $\norm{\vG^{(u_m)}}$ is minimized over all $m \in \cM_{s, e}$ and say the minimization occurs at $m^*$. Then $u_0 := u_{m^*}$ is declared as a change-point.
	\item Algorithm again recursively starts for the two intervals $(s, u_0)$ and $(u_0+1, e)$.
\end{enumerate}

\vspace{0.2in}
\framebox[\textwidth]
{\centering\parbox{.95\textwidth}{
		\textbf{Algorithm 2:} Wild Binary Segmentation $(s, e, \gk)$ \\
		\textbf{Input:} Adjacency matrices $\vA^{(s)}, \vA^{(s+1)}, \ldots, \vA^{(e)}$; cushion $\gk$.\\
		\textbf{Output:} Change point estimate $\hat\tau$. \\
		\begin{enumerate}
			\item If $ e - s < \gk $ then
			\item \quad STOP
			\item else
			\item \quad $\cM_{s,e} :=$ set of those indices $m$ for which $[s_m, e_m] \in \cF_T^M$ is such that $[s_m, e_m]\subseteq [s,e]$.
			\item \quad Define $\gL = e_m - s_m$.
			\item \quad If $e_m -s_m < \gk$, then,
			\item \quad \quad CONTINUE
			\item \quad else
			\item \quad \quad For $i=1, 2, \ldots, n$ do
			\item \quad \quad\quad\quad Obtain degree of each vertex, $\{D_{i,m}\}_{i=1}^n$.
			\item \quad \quad Order the values $D_{1,m}, \ldots, D_{n,m}$ to get 
			$D_{(1),m}\le \cdots\le D_{(n),m}$.
			\item \quad\quad  Obtain row indices $i_1, \ldots, i_\gC$ such that $D_{i_k,m}\ge D_{(n+1-\gC),m}$
			\item \quad\quad  Obtain $\tilde \vA^{(s)}$ from $\vA^{(s)}$ for each $s\in(s_m, e_m)$ by removing rows and columns with indices $i_1, \ldots, i_\gC$. 
			\item \quad \quad For $t=s_m + \frac 13 \gk, s_m + \frac 13 \gk +1, \ldots, e_m - \frac 13 \gk$ do
			\item \quad\quad\quad \quad  Obtain $c^{(t)}_m:=\sqrt{\frac{t}{(e_m - s_m)}\left(1 - \frac{t}{(e_m - s_m)}\right)}$.
			\item \quad\quad\quad \quad  Obtain $\vG^{(t)}_m:=c^{(t)}_m\times \left[\frac 1t\sum_{s\in(s_m,t]}\tilde\vA^{(s)} - \frac{1}{e_m - t}\sum_{s\in(t,e_m]}\tilde\vA^{(s)}\right]$.
			\item \quad \quad Obtain $u_m=\arg\max_{t\in(s_m + \frac 13 \gk, e_m - \frac 13 \gk)} \norm{\vG^{(t)}_m}$
			\item \quad\quad  If $\norm{\vG_m^{(u_m)}}>\Theta_\mu\left[\frac{\bar D}{\gk}
			\left(\zeta+6+\frac{\log(M)}{\log(n)}\right)\right]^{1/2}$, declare $u_m$ as a \emph{candidate change point}.  
			\item \quad end if
			\item \quad $u_0=\arg\max_{m \in \cM_{s,e}} \norm{\vG^{(t)}_m}$
			\item \quad Add $u_0$ to the set of estimated change-points
			\item \quad Wild Binary Segmentation $(s, u_0, \gk)$
			\item \quad Wild Binary Segmentation $(u_0 + 1, e, \gk)$
			\item \quad else
			\item \quad \quad STOP
			\item \quad end if
			\item end if
		\end{enumerate}
}}\\

\section{Theory}
\label{sec:theory}

\subsection{Lower Bound}
\label{sec_lower_bound}
Let $\pr_\Theta$ denote the distribution of the inhomogeneous random graph having mean adjacency matrix $\Theta$.
For $\rho, \ga\in(0,1), \gk\in[T]$,  and (possibly random) symmetric matrix $\gC$ taking values in $ [-1, 1]^{n\times n}$, let 
\begin{align}
& \Theta_0=\rho\ol{\mathbf 1_n\mathbf 1_n^T}, 
\pr_0^\gk:=\pr_{\Theta_0}\overbrace{\times\cdots\times}^{\gk \text{ times}}\pr_{\Theta_0},  \Theta(\gC,\ga)=\Theta_0+\ga\rho\ol{\gC}, \label{P kappa def}\\
& \pr_{\gC,\ga}^\gk:=\pr_{\Theta(\gC,\ga)}\overbrace{\times\cdots\times}^{\gk \text{ times}}\pr_{\Theta(\gC,\ga)}, 
\pr_{1,\ga}^\gk(\cdot)=\E_\gC \pr_{\gC,\ga}^\gk(\cdot), \notag
\end{align}
where $\pr_\gC$ denotes the distribution of $\gC$.
\begin{lem}[Chi-square divergence bound] \label{ChiSqBd}
	Let $\bre\pr_\gC$ denote the distribution of $\gC\in\dS_n$, where $(\gC_{ij}, 1\le i<j\le n)$ are \iid with finite second moment. For any $\dB\subset\dS_n$ satisfying $\bre\pr_\gC(\dB)\ge 1/2$, let $\pr_\gC(\cdot)=\bre\pr_\gC(\cdot| \dB)$.
	For any $\eps>0$, there is a constant $c(\eps)>0$ such that $\chi^2(\pr_0^\gk, \pr_{1,\ga}^\gk)\le \eps$ whenever $\gk\ga^2\rho n \le c(1-\rho)$.
\end{lem}

\begin{proof}
	Let  $\tilde\gC$ be an independent copy of $\gC$, $\pr_{\gC\otimes\tilde\gC}$ denote their joint distribution, and $\gs^2:=\E_\gC(\gC_{ij}^2)$. 
	Noting  that $(A^{(t)}_{ij}, 1\le i<j\le n, t\in[\gk])$ are \iid $Ber(\rho)$ under $\pr_0^\gk$, and using the inequality $1+x\le e^x$ for any $x\in\R$, 
	\begin{align}
	1+\chi^2\left(\pr_0^\gk, \pr_{1,\ga}^\gk\right)
	=\E_0^\gk\left[\left(\frac{d\pr_{1,\ga}^\gk}{d\pr_0^\gk}\right)^2\right] 
	=\E_0^\gk\left[\E_{\gC\otimes\tilde\gC}\left(\frac{d\pr_{\gC,\ga}^\gk}{d\pr_0^\gk}\frac{d\pr_{\tilde\gC,\ga}^\gk}{d\pr_0^\gk}\right)\right]  \label{chi^2bd} \\ 
	=\E_{\gC\otimes\tilde\gC} \E_0^\gk\prod_{\substack{1\le i< j\le n\\ t\in[\gk]}}
	\left((1+\ga\gC_{ij})(1+\ga\tilde\gC_{ij}) \mathbf 1_{\{A^{(t)}_{ij}=1\}} + \left(1-\frac{\ga\rho\gC_{ij}}{1-\rho}\right) 
	\left(1-\frac{\ga\rho\tilde\gC_{ij}}{1-\rho}\right) \right. \notag \\
	\left.\mathbf 1_{\{A^{(t)}_{ij}=0\}}\right) = \E_{\gC\otimes\tilde\gC}\prod_{\substack{1\le i<j\le n\\t\in[\gk]}} \left(1+\frac{\ga^2\rho\gC_{ij}\tilde\gC_{ij}}{1-\rho}\right)
	\le  \E_{\gC\otimes\tilde\gC}\exp\left[\frac{\gk\ga^2\rho F(\gC,\tilde\gC)}{1-\rho}\right], \notag
	\end{align}
	where $F(\gC,\tilde\gC):=\sum_{1\le i<j\le n} \gC_{ij}\tilde\gC_{ij}$.
	Using Bernstein inequality and the fact that $\bre\pr(\dB)\ge 1/2$,
	\beqax
	\pr_{\gC\otimes\tilde\gC}\left(F(\gC,\tilde\gC) >\ell\gs^2 n\right) \le 4\bre\pr_{\gC\otimes\tilde\gC}\left(F(\gC,\tilde\gC) >\ell\gs^2 n\right)  \\
	\le 4\exp\left(\frac{-\frac 12\ell^2n^2\gs^4}{\sum_{1\le i<j\le n}\E_\gC(\gC_{ij}^2)\E_{\tilde\gC}(\tilde\gC_{ij}^2)+\frac 23 \ell \gs^2n}\right) \le e^{-3\ell^2/7}.
	\eeqax
	For any $L\ge 1$, we can use the above estimate to have 
	\begin{align*}
	& \E_{\gC\otimes\tilde\gC}\exp\left[\frac{\gk\ga^2\rho}{1-\rho} F(\gC,\tilde\gC)\right]
	\le  \E_{\gC\otimes\tilde\gC}\left(\exp\left[\frac{\gk\ga^2\rho}{1-\rho} F(\gC,\tilde\gC)\right]
	\mathbf 1_{\{F\le L\gs^2n\}}\right) \\
	&+ \sum_{\ell=L}^\infty \E_{\gC\otimes\tilde\gC}\left(\exp\left[\frac{\gk\ga^2\rho}{1-\rho} F(\gC,\tilde\gC)\right]
	\mathbf 1_{\{\ell\gs^2n<F\le (\ell+1)\gs^2n\}}\right)
	\le \exp\left[\frac{\gk\ga^2\rho}{1-\rho} L\gs^2n\right] \\
	&+ \sum_{\ell=L}^\infty \exp\left[\frac{\gk\ga^2\rho}{1-\rho} (\ell+1)\gs^2n\right] \pr_{\gc\otimes\tilde\gC}(F>\ell\gs^2n) \\
	&\le \exp\left[\frac{\gk\ga^2\rho}{1-\rho} L\gs^2n\right] 
	+ \sum_{\ell=L}^\infty 4\exp\left[\frac{\gk\ga^2\rho}{1-\rho} (\ell+1)\gs^2n-\frac 37\ell^2\right].
	\end{align*}
	Given $\eps>0$, we can choose $L(\eps)$ large enough so that $4\sum_{\ell\ge L} e^{(\ell+1)-3\ell^2/7} \le \eps/2$.
	Having chosen $L(\eps)$, we can choose $c(\eps)>0$ small enough such that $e^{c(\eps)L(\eps)\gs^2}\le 1+\eps/2$. Combining this with \eqref{chi^2bd} proves the result.
\end{proof}

\noindent
For $\gTh_0, \gTh_1\in\dS_n$ and $\tau\in[T]$, let  $\dS_n^{T,\gt}:=\{\bsQ\in\dS_n^T: \cK(\bsQ)=1, \cT_1(\bsQ)=\gt\}$,
\[ \bsQ(\gt,T;\gTh_0,\gTh_1)=\left(\overbrace{\gTh_0,\ldots, \gTh_0}^{\gt \text{ many}},\overbrace{\gTh_1,\ldots, \gTh_1}^{T-\gt \text{ many}}\right)\in\dS_n^{T,\gt}. \]
For any (possibly degenerate) probability distributions $\gN_0, \gN_1$ on $\dS_n$, let 
\beqx
\pr_{\bsQ(\gt,T;\gN_0,\gN_1)}(\cdot) := \E_{(\gTh_0,\gTh_1)\sim\gN_0\otimes\gN_1} \pr_{\bsQ(\gt,T;\gTh_0,\gTh_1)}(\cdot), \text{ and }\bsQ(\gt,T;\gN_0,\gN_1) 
\eeqx
be the distribution on $\dS_n^{T,\gt}$ satisfying
$\pr(\bsQ(\gt,T;\gN_0,\gN_1)=\bsQ(\gt,T;\gTh_0,\gTh_1)) = \gN_0\otimes\gN_1(\gTh_0,\gTh_1), $. For $r\in[n], \gk\in[T]$, and $\gc>0$, let
\beqax \sH(\gc, r, \gk) := \left\{\pr_{\bsQ(\gt,T;\gN_0,\gN_1)}:  (a) \;\gN_0 \text{ and  }  \gN_1 \text{ are distributions on } \dS_n,  \right. \\ 
(b) \;(\Theta_0,\Theta_1)\sim\gN_0\otimes\gN_1
\text{implies } \norm{\gTh_1-\gTh_0}_F^2 \le \frac{\gc}{\gk} n(\norm{\Theta_0}_{max} \vee \norm{\Theta_1}_{max}),  \\
\left.  \text{and rank}(\gTh_0-\gTh_1)=r, \text{ and }(c) \;\gk\le\gt\le T-\gk\right\}
\eeqax
\begin{lem}[Lower bound for localizability]\label{lower_local}
	For any $\eps>0$, there is a constant $\gc(\eps)>0$ such that for any $r\in[n]$ and $\gk< T/2$,
	$ \inf_{\hat\tau}\sup_{\pr_{\bsQ}\in\sH(\gc, r, \gk)}\E_{\bsQ}|\hat\tau-\cT_1(\bsQ)|\ge (T-2\gk)(1-2\eps)$.
\end{lem}

\begin{proof}
	Given $\eps>0$, take $\gc$ to be equal to $c(\eps)$ (the constant in Lemma \ref{ChiSqBd}).
	Let (a) $\ga,\rho\in(0,1)$ be such that $\gk\ga^2\rho n\le \gc(1-\rho)$, (b) $\gTh_0=\ol{\rho\mathbf 1_n\mathbf 1_n^T}$, (c) $\gC=\sum_{i=1}^r  3^{-i}\vU_i\vU_i^T$, where $\vU_1, \ldots, \vU_r$ are \iid $n\times 1$ vectors with $\pr_\gC(U_{ij}=\pm 1)=1/2$, (d) $\gN_1$ be the distribution of $\gTh_{\gC,\ga}:=\gTh_0+\ol{\ga\rho\gC}$ given the event $\dB:=\{\text{rank}(\gC)=r\}$, (e) $\bsQ^{[0]}:=\bsQ(\gk,T;\gN_1,\gN_0)$, (f) $\bsQ^{[1]}:=\bsQ(T-\gk,T;\gN_0,\gN_1)$, and (g) $\pr_0^\gk, \pr_{1,\ga}^\gk$ are as in \eqref{P kappa def}.
	It is well known that the probability of $\dB$ goes to 1 as $n\to\infty$ \cite{}.
	It is easy to see that if $\gTh_1\sim\gN_1$, then
	$\norm{\gTh_1-\gTh_0}_F^2 \le \ga^2\rho^2n^2 \le \gc \rho n/\gk$. Thus $\pr_{\bsQ*{[i]}}\in\sH(\gc,r,\gk)$ for both $i=1,2$, as $||\gTh_0||_{max}=\rho$. Also, $|\cT_1(\bsQ^{[0]}) - \cT_1(\bsQ^{[1]})| = T-2\gk$. So, using Le Cann's lemma \cite{},
	\[ \inf_{\hat\tau}\sup_{\pr_{\bsQ}\in\sH(\gc, r,\gk)} \E_{\bsQ}|\hat\tau - \cT_1(\bsQ)| \ge (T-2\gk) (1-d_{TV}(\pr_{\bsQ^{[0]}}, \pr_{\bsQ^{[1]}}),\]
	where $d_TV(\cdot,\cdot)$ denotes the total variation distance. 
	Using the facts that $d_TV(\cdot,\cdot)\le \chi^2(\cdot,\cdot)$ and $d_TV(\gM_1\otimes\gM_2,\gN_1\otimes\gN_2) \le d_{TV}(\gM_1,\gN_1) + d_{TV}(\gM_2,\gN_2)$ for all compatible probability distributions $\gM_i,\gN_i$, we see that
	\[ d_{TV}(\pr_{\bsQ^{[0]}}, \pr_{\bsQ^{[1]}}) 
	\le 2d_{TV}(\pr_{0}^\gk, \pr_{1,\ga}^\gk) \le \chi^2(\pr_{0}^\gk, \pr_{1,\ga}^\gk),\]
	which is at most $2\eps$ by Lemma \ref{ChiSqBd} and the choice of $\ga$.
\end{proof}

\begin{lem}[Lower bound for detectability]\label{lower_detect}
	For any $\eps>0$, there is a constant $c(\eps)>0$ such that if $\dY:=\{\bsQ\in\dS_n^T:\cK(\bsQ)\ge 1, \fS(\bsQ)\le c(\eps)\sqrt{\frac{\cD(\bsQ)}{\fg(\bsQ)}}\}$, then $\Pi(\gO_{n,T};\dS_{n,0}^T,\dY) \ge 1-\sqrt\eps/2$ for all test function $\gO_{n,T}$.
\end{lem}

\begin{proof}
	Fix any $\rho\in(0,1)$. Given $\eps>0$ let $c(\eps)$ be the constant of Lemma \ref{ChiSqBd}. Choose $\ga$ so that $\gk\ga^2\rho n\le c(1-\rho)$.
	Let  $\gTh_0=\ol{\rho\mathbf 1_n\mathbf 1_n^T}, \gTh(\vU,\ga)=\gTh_0+\ga\ol{\vU\vU^T},$ where $\vU\in\R^{n\times 1}$ has \iid components with $\pr(U_i=\pm 1)=1/2$,
	\begin{align*}
	\bsQ_0:=\left(\gTh_0, \ldots, \gTh_0\right), \bsQ(\vU,\ga)=\left(\overbrace{\gTh(\vU,\ga), \ldots, \gTh(\vU,\ga)}^{\gk \text{ times}},  \gTh_0, \ldots, \gTh_0\right)
	\in\dS_n^T
	\end{align*}
	Let $\pr_0$ and $\pr_1$ represent the distributions $\pr_{\bsQ_0}$ and $\E_\vU\pr_{\bsQ(\vU,\ga)}$ respectively. Let $\gO_{n,T}^*:=\mathbf 1_{\{\frac{d\pr_1}{d\pr_0}>1\}}$
	Then, for any test function $\gO_{n,T}$, 
	\begin{align*}
	\Pi\left(\gO_{n,T}; \dS_{n,0}^T,\dY\right)
	\ge \pr_0(\gO_{n,T}=1)+\pr_1(\gO_{n,T}=0) \\
	= 1+\int_{\{\gO_{n,T}=1\}}\left[1-\frac{d\pr_1}{d\pr_0}\right]d\pr_0
	\ge 1+\int_{\{\gO_{n,T}^*=1\}}\left[1-\frac{d\pr_1}{d\pr_0}\right]d\pr_0 \\
	=1-\frac 12 \E_0\left|\frac{d\pr_1}{d\pr_0}-1\right|
	\ge 1-\frac 12 \sqrt{\E_0\left[\left(\frac{d\pr_1}{d\pr_0}\right)^2\right]-1} \ge 1-\sqrt\eps/2
	\end{align*}
	by the choice of $\ga$ and the result of  Lemma \ref{ChiSqBd}.
\end{proof}

\subsection{Multiple Change Points}
\label{sec:multiple}
Algorithm 1 and 2 were presented in \textsection \ref{sec_multiple} to detect and localize multiple change points in network sequences. Recall that the data is given in form of a sequence of networks represented by the symmetric adjacency matrices $\left(\vA^{(1)}_{n\times n}, \ldots, \vA^{(T)}_{n\times n}\right)$. The generating model (presented in \textsection \ref{sec_multiple}) with change points $(\tau_1, \tau_2, \ldots, \tau_K)$ and parameters $\bsQ = \left(\vQ^{(1)}, \ldots, \vQ^{(T)}\right)$, can be recalled as -
\begin{enumerate}[(i)]
	\item if $\vQ^{(t)} = \vQ$, for all $t \in [T]$, there exists no change-point;
	\item if there exists a set of change-points $(\tau_1, \tau_2, \ldots, \tau_K)$ with $\tau_0 := 0  \le \tau_1 \le \cdots \le \tau_K \le T =: \tau_{K+1}$ such that 
	\[\vQ^{(t)} = \vQ_k \ \  \text{for } \tau_{k-1} \le t \le \tau_k,\ k\in [K]\]
	where, $\vQ_k \in \dS_n$ for $k\in [K]$, then there exists a multiple change-points $(\tau_1, \tau_2, \ldots, \tau_K)$. 
	\item Additionally $\fg(\bsQ)$ represents the \emph{minimum cushion} as defined in \eqref{cushion_def}, $\fg(\bsQ) := \min_{0\le k\le K} (\tau_{k+1}-\tau_k)$.	
	\item \emph{Signal strength} is given by 
	\begin{align}\label{eq_signal}
	\fS(\bsQ) := \begin{cases}
		\min_{k\in[\cK(\bsQ)]} \norm{\vQ^{(\tau_k+1)}-\vQ^{(\tau_k)}} & \text{ if } \cK(\bsQ)\ge 1\\
		0 & \text{ if } \cK(\bsQ)=0\end{cases}
	\end{align}
	\item \emph{Sparsity} parameter is given by 
	\begin{align}\label{eq_sparsity}
	\cD(\bsQ):=\max_{i,j\in[n], t\in[T]} Q^{(t)}_{ij}.
	\end{align}
\end{enumerate}

\subsubsection{Window-based Method}
\label{sec_window_theory}
Algorithm 1, which is presented in \textsection \ref{sec_multiple_window}, outputs $(\hat \tau, \gL, \cI)$, which are the estimates of the change point $\tau$ along with interval length $\gL$, and interval $\cI$ containing the estimated change point. The theoretical results on the performance of the change point estimate $\hat \tau$ is given in Theorems \ref{multiple tau detect} and \ref{multiple tau local}.  
We show two different types of theoretical results for the change point detection problem:
\begin{enumerate}
	\item \emph{Detectability:} The results on detectability focuses on correctly detecting the presence or absence of change point in a network sequence. The loss function for detection is given in terms of minimax loss as defined in \eqref{minimax_risk}. The theoretical result on detectability for window-based algorithm (Algorithm 1) is given in Theorem \ref{multiple tau detect}.
	\item \emph{$\gL$-localizability:} The results on $\gL$-localizability focuses on correctly estimating locations of all the change points in a network sequence and giving an interval estimate of length $\gL$ around the true change points. The loss function for $\gL$-localizability is given in form of the error made in estimating the location of the change point. The theoretical result on $\gL$-localizability for window-based algorithm (Algorithm 1) is given in Theorem \ref{multiple tau local} and for wild binary segmentation algorithm (Algorithm 2) is given in Theorem \ref{thm_wbs}.
\end{enumerate}
\begin{thm}[Detectability Result] \label{multiple tau detect}
	Let us consider that we have a sequence of networks represented by the symmetric adjacency matrices $\bsA = \left(\vA^{(1)}_{n\times n}, \ldots, \vA^{(T)}_{n\times n}\right)$ generated from the MIRGraM model with parameters $\bsQ = (\vQ^{(1)}, \ldots, \vQ^{(T)}) \in \dS_n^T$ with the set of change points given by $\cT(\bsQ) := (\tau_1, \tau_2, \ldots, \tau_K)$ if $\cK(\bsQ) > 0$. Then, we have the following results on detecting change points in the sequence $\bsA$.
	\begin{enumerate}
		\item \emph{(Lower bound)} For any $\eps>0$, there is a constant $c(\eps)>0$ such that if $\dY:=\left\{\bsQ\in\dS_n^T:\cK(\bsQ)\ge 1, \fS(\bsQ)\le c(\eps)\sqrt{\frac{\cD(\bsQ)}{\fg(\bsQ)}}\right\}$, then $\Pi(\gO_{n,T};\dS_{n,0}^T,\dY) \ge 1-\sqrt\eps/2$ for all test function $\gO_{n,T}$.
		\item \emph{(Upper bound)} Consider that Algorithm 1 is applied on $\bsA$ with cushions of length $\gk$ ($3\le \gk \le T$), and intervals of length $\gL$ ($3 \le \gL\le \gk$). There are constants $C_1, c_1, \zeta_0>0$ with $\Psi_\mu$ as in \eqref{Psi def_gen} for $\mu>0$, $d$ and $\fd$ as in \eqref{d def} such that if $\dW$ be the set $\bsQ \in \dS_n^T$ satisfying the properties $\cK(\bsQ) \ge 1$, and
		\begin{align}
		\label{eq_thm1_cond}
		\gk > \fg(\bsQ) \text{ and } \fS(\bsQ) \geq (C_1+c_1\Psi_\mu)\left[\frac{d}{\gL\wedge\gk}\left(\zeta+\frac{\log(|\cL|)}{\log(n)}\right)\right]^{1/2}
		\end{align}  
		then,
		\begin{align*}
		\Pi(\gO_{n,T};\dS_{n,0}^T,\dW) \le 2\frac{(|\cL|)^{9/\log(n)}}{\log(n)} n^{-\zeta}+3\exp\left[-\mu(\gL\wedge\gk)\fd\right] \text{ for all } \zeta>\zeta_0
		\end{align*}
		
	\end{enumerate}	
\end{thm}
\begin{proof}
	\begin{enumerate}
		\item The proof follows from Lemma \ref{lower_detect}.
		\item The proof follows from the result in Theorem \ref{multiple tau local}.
	\end{enumerate}
\end{proof}
\begin{rem}
	Note that in Algorithm 1, $\gL \leq \gk$, so $\gL\wedge \gk = \gL$ in Theorem \ref{multiple tau detect}. 
\end{rem}

\begin{thm}[$\gL$-localizability Result] \label{multiple tau local}
	Let us consider that we have a sequence of networks represented by the symmetric adjacency matrices $\bsA = \left(\vA^{(1)}_{n\times n}, \ldots, \vA^{(T)}_{n\times n}\right)$ generated from the MIRGraM model with parameters $\bsQ = (\vQ^{(1)}, \ldots, \vQ^{(T)}) \in \dS_n^T$ with the set of change points given by $\cT(\bsQ) := (\tau_1, \tau_2, \ldots, \tau_K)$ if $\cK(\bsQ) > 0$. Then, we have the following results on $\gL$-localizability of change point estimates in the sequence $\bsA$.
	\begin{enumerate}
		\item \emph{(Lower bound)} For the case of single change point $\tau\in[T]$, let $\dS_n^{T,\gt}:=\{\bsQ\in\dS_n^T: \cK(\bsQ)=1, \cT_1(\bsQ)=\gt\}$, and 
		\[ \bsQ(\gt,T;\gTh_0,\gTh_1)=\left(\overbrace{\gTh_0,\ldots, \gTh_0}^{\gt \text{ many}},\overbrace{\gTh_1,\ldots, \gTh_1}^{T-\gt \text{ many}}\right)\in\dS_n^{T,\gt}, \]
		where, $\gTh_0, \gTh_1\in\dS_n$. For any (possibly degenerate) probability distributions $\gN_0, \gN_1$ on $\dS_n$, let 
		\beqx
		\pr_{\bsQ(\gt,T;\gN_0,\gN_1)}(\cdot) := \E_{(\gTh_0,\gTh_1)\sim\gN_0\otimes\gN_1} \pr_{\bsQ(\gt,T;\gTh_0,\gTh_1)}(\cdot), \text{ and }\bsQ(\gt,T;\gN_0,\gN_1) 
		\eeqx
		be the distribution on $\dS_n^{T,\gt}$ satisfying $\pr(\bsQ(\gt,T;\gN_0,\gN_1)=\bsQ(\gt,T;\gTh_0,\gTh_1)) = \gN_0\otimes\gN_1(\gTh_0,\gTh_1)$. For $r\in[n], \gk\in[T]$, and $\gc>0$, let
		\beqax \sH(\gc, r, \gk) := \left\{\pr_{\bsQ(\gt,T;\gN_0,\gN_1)}:  (a) \;\gN_0 \text{ and  }  \gN_1 \text{ are distributions on } \dS_n,  \right. \\ 
		(b) \;(\Theta_0,\Theta_1)\sim\gN_0\otimes\gN_1
		\text{implies } \norm{\gTh_1-\gTh_0}_F^2 \le \frac{\gc}{\gk} n(\norm{\Theta_0}_{max} \vee \norm{\Theta_1}_{max}),  \\
		\left.  \text{and rank}(\gTh_0-\gTh_1)=r, \text{ and }(c) \;\gk\le\gt\le T-\gk\right\}
		\eeqax
		Then, for any $\eps>0$, there is a constant $\gc(\eps)>0$ such that for any $r\in[n]$ and $\gk< T/2$,
		\[ \inf_{\hat\tau}\sup_{\pr_{\bsQ}\in\sH(\gc, r, \gk)}\E_{\bsQ}|\hat\tau-\cT_1(\bsQ)|\ge (T-2\gk)(1-2\eps).\]
		\item \emph{(Upper bound)} Consider that Algorithm 1 is applied on $\bsA$ with cushions of length $\gk$ ($3\le \gk \le T$), and intervals of length $\gL$ ($3 \le \gL\le \gk$). There are constants $C_1, c_1, \zeta_0>0$ with $\Psi_\mu$ as in \eqref{Psi def_gen} for $\mu>0$, $d$ and $\fd$ as in \eqref{d def} such that if $\dW$ be the set $\bsQ \in \dS_n^T$ satisfying the properties $\cK(\bsQ) \ge 1$, and
		\begin{align}
		\label{eq_thm2_cond}
		\gk > \fg(\bsQ) \text{ and } \fS(\bsQ) \geq (C_1+c_1\Psi_\mu)\left[\frac{d}{\gL\wedge\gk}\left(\zeta+\frac{\log(|\cL|)}{\log(n)}\right)\right]^{1/2}
		\end{align}  
		then,
				\beqax
				\pr\left(|\hat{\tau_i} - \tau_i| \ge \gL \frac{C_1+c_1\Psi_\mu}{\fS(\bsQ)}\left[\frac{d}{\gL\wedge\gk}\left(\zeta+\frac{\log(|\cL|)}{\log(n)}\right)\right]^{1/2}\;\forall\; i\in [K] \text{ and } \hat K=K\right) \\
				\le  2\frac{(|\cL|)^{9/\log(n)}}{\log(n)} n^{-\zeta}+3\exp\left[-\mu(\gL\wedge\gk)\fd\right] \text{ for all } \zeta>\zeta_0. 
				\eeqax
	\end{enumerate}	
\end{thm}
\begin{proof}
	\begin{enumerate}
		\item The proof follows from Lemma \ref{lower_local}.
		\item The proof is in the Appendix.
	\end{enumerate}
\end{proof}

\begin{rem}
	Note that in Algorithm 1, $\gL \leq \gk$, so $\gL\wedge \gk = \gL$ in Theorem \ref{multiple tau local}. Also, if $\fg(\bsQ) \le \gk = O(T)$ and $\gL = O(T)$, then, under the condition weaker condition of $\fS(\bsQ) \ge C\sqrt{\frac{d}{T}}$ (for some constant $C>0$), $|\hat{\tau} - \tau| = O(T)$ with high probability. However, if $\fg(\bsQ) \leq \gk = o(T)$ and $\gL = o(T)$, then, under the stronger condition of $\fS(\bsQ) \ge C\sqrt{\frac{d}{\gL}}$ (for some constant $C>0$), $|\hat{\tau} - \tau| = O(\gL)$ with high probability. So, as the signal strength $\fS(\bsQ)$ increases, the interval length of change point estimate, $\gL$, decreases, that is, the localization of the change point estimate becomes better.
\end{rem}

%

\subsubsection{Wild Binary Segmentation}
\label{sec:wbs}
Algorithm 2, which is presented in \textsection \ref{sec_wbs}, outputs $(\hat \tau, \gL, \cI)$, which are the estimates of the change point $\tau$ along with interval length $\gL$, and interval $\cI$ containing the estimated change point. The generating model is presented in \textsection \ref{sec_multiple}. The theoretical results on the performance of the change point estimate $\hat \tau$ is given in Theorems \ref{thm_wbs}.  
\begin{thm} \label{thm_wbs}
	Consider that Algorithm 2 is applied on $\bsA$ with cushions of length $\gk$ ($3\le \gk \le T$). There are constants $C_1, c_1, \zeta_0>0$ with $\Psi_\mu$ as in \eqref{Psi def_gen} for $\mu>0$, $d$ and $\fd$ as in \eqref{d def} such that if $\dW$ be the set $\bsQ \in \dS_n^T$ satisfying the properties $\cK(\bsQ) \ge 1$, and
	\begin{align}
	\label{eq_thm_wbs_cond}
	\gk > \fg(\bsQ) \text{ and } \fS(\bsQ) \geq (C_1+c_1\Psi_\mu)\left[\frac{d}{\gk}\left(\zeta+\frac{\log(M)}{\log(n)}\right)\right]^{1/2}
	\end{align}  
	then, for all $\zeta>\zeta_0$,
	\beqax
	\pr\left(|\hat{\tau_i} - \tau_i| \ge \gL \frac{C_1+c_1\Psi_\mu}{\fS(\bsQ)}\left[\frac{d}{\gk}\left(\zeta+\frac{\log(M)}{\log(n)}\right)\right]^{1/2}\;\forall\; i\in [K] \text{ and } \hat K=K\right) \\
	\le  2\frac{(M)^{9/\log(n)}}{\log(n)} n^{-\zeta}+3\exp\left[-\mu\gk\fd\right] + T\gk^{-1}\left(1 - \gk^2T^{-2}/9\right)^M. 
	\eeqax
\end{thm}	

\begin{proof}
	The proof is in the Appendix.
\end{proof}

\bibliographystyle{imsart-nameyear}
\bibliography{extracted}

\end{document}